\documentclass[12pt,onecolumn,draftclsnofoot,journal]{IEEEtran}
\usepackage{graphicx}
\usepackage{float}
\usepackage{epstopdf}
\usepackage[cmex10]{amsmath}
\usepackage{array}
\usepackage{cite}
\usepackage{amssymb}
\usepackage{amsfonts}
\usepackage{amsmath}
\usepackage{stackrel}
\usepackage{booktabs,multirow}
\usepackage{arydshln}
\usepackage{color}
\usepackage{amsthm}
\usepackage{listings}
\usepackage{algorithm}
\usepackage{algorithmicx}
\usepackage{algpseudocode}
\usepackage{threeparttable}
\newcommand\numberthis{\addtocounter{equation}{1}\tag{\theequation}}

\newtheorem{lem}{Lemma}

\newtheorem{theo}{Theorem}

\newtheorem{ex}{Example}
\newfloat{routine}{htbp}{loa}
\floatname{routine}{Routine}

\makeatletter
\newcommand{\algmargin}{\the\ALG@thistlm}
\makeatother
\newlength{\forwidth}
\algdef{SE}[parFOR]{parFor}{EndparFor}[1]
{\parbox[t]{\dimexpr\linewidth-\algmargin}{
		\hangindent\forwidth\strut\algorithmicfor\ #1\ \algorithmicdo\strut}}{\algorithmicend\ \algorithmicfor}
\algnewcommand{\parState}[1]{\State
	\parbox[t]{\dimexpr\linewidth-\algmargin}{\strut #1\strut}}

\newlength{\ifwidth}
\algdef{SE}[parIF]{parIf}{EndparIf}[1]
{\parbox[t]{\dimexpr\linewidth-\algmargin}{
		\hangindent\ifwidth\strut\algorithmicif\ #1\ \algorithmicdo\strut}}{\algorithmicend\ \algorithmicif}

\begin{document}

\title{On Computing the Number of  Short Cycles in Bipartite Graphs Using the Spectrum of the Directed Edge Matrix}
\author{Ali Dehghan, and Amir H. Banihashemi,\IEEEmembership{ Senior Member, IEEE}}

\maketitle


\begin{abstract}
Counting short cycles in bipartite graphs is a fundamental problem of interest in many fields including the analysis and design of low-density parity-check (LDPC) codes. There are two computational approaches to count short cycles (with length smaller than $2g$, where $g$ is the girth of the graph) in bipartite graphs. The first approach is applicable to a general (irregular) bipartite graph, and uses the spectrum $\{\eta_i\}$ of the directed edge matrix of the graph to compute the multiplicity $N_k$ of $k$-cycles with $k < 2g$ through the simple equation $N_k = \sum_i \eta_i^k/(2k)$.
This approach has a computational complexity $\mathcal{O}(|E|^3)$, where $|E|$ is number of edges in the graph.  The second approach is only applicable to bi-regular bipartite graphs, and uses the  spectrum $\{\lambda_i\}$ of the adjacency matrix (graph spectrum) and the degree sequences of the graph to compute $N_k$. The complexity of this approach is $\mathcal{O}(|V|^3)$, where $|V|$ is number of nodes in the graph. This complexity is less than that of the first approach, but the equations involved in the computations of the second approach are very tedious, particularly for $k \geq g+6$. In this paper, we establish an analytical relationship between the two spectra $\{\eta_i\}$ and $\{\lambda_i\}$ for bi-regular bipartite graphs. Through this relationship, the former spectrum can be derived from the latter through simple equations. This allows the computation of $N_k$ using $N_k = \sum_i \eta_i^k/(2k)$ but with a complexity of $\mathcal{O}(|V|^3)$ rather than $\mathcal{O}(|E|^3)$.

\begin{flushleft}

\end{flushleft}

\end{abstract}

\section{introduction}
\label{section01}

Bipartite graphs appear in many fields of science and engineering to represent systems that are described by local constraints on different subsets of variables involved in the description of the system. In such a representation, the nodes on one side of the bipartition represent the variables while the nodes on the other side are representative of the constraints. One example is the Tanner graph representation of low-density parity-check (LDPC) codes, where variable nodes represent the code bits and the constraints are parity-check equations. In the bipartite graph representation of systems, the cycle distribution of the graph often plays an important role in understanding the properties of the system. For example, the performance of LDPC codes, both in waterfall and error floor regions,  is highly dependent on the distribution of short cycles of the Tanner graph~\cite{mao2001heuristic},~\cite{hu2005regular},~\cite{halford2006algorithm},~\cite{xiao2009error},~\cite{MR3071345},~\cite{asvadi2011lowering},~\cite{MR2991821},~\cite{MR3252383},~\cite{HB-CL},~\cite{HB-IT1}.
 
Motivated by this,  in the coding community, there has been a large body of work on the distribution and counting of cycles in bipartite graphs, see, e.g., \cite{halford2006algorithm}, \cite{karimi2012counting}, \cite{karimi2013message}, \cite{dehghan2016new},  \cite{blake2017short}. 

Generally, counting cycles of a given length in a given graph is known to be NP-hard \cite{flum2004parameterized}. The problem remains NP-hard even for the family of bipartite graphs \cite{MR1405031}. There are, in general, two computational approaches to count the number of short cycles in bipartite graphs. The first approach is applicable to any (irregular) bipartite graph, and is described in the following theorem.

\begin{theo}\label{T02} \cite{karimi2012counting}
Consider a bipartite graph $G$ with the directed edge matrix $A_e$, and let $\{\eta_i\}$ be the spectrum of $A_e$. Then, the number of $k$-cycles in $G$ is given by $N_k= \dfrac{\sum_{i} \eta_i^k}{2k}$, for $k < 2g$, where $g$ is the girth of $G$.
\end{theo}

The result of Theorem~\ref{T02} follows from the property of $A_e$ that the
number of tailless backtrackless closed (TBC) walks of length $k$ in $G$ 
is equal to $tr(A_e^k)/2k$, where $tr(A_e)$ denotes the trace of $A_e$. This together
with the fact that the set of TBC walks of length less than $2g$ coincides with the set of cycles of the 
same size~\cite{karimi2013message} prove the result. To use Theorem \ref{T02}, one needs 
to calculate the eigenvalues of $A_e$. This has a complexity of $\mathcal{O}(|E|^3)$, where $|E|$ is number of edges in the graph~\cite{MR2460593}.

The second approach, which was introduced by Blake and Lin~\cite{blake2017short} and extended by Dehghan and Banihashemi \cite{dehghan2018spectrum},
uses the spectrum of the adjacency matrix and the degree distribution of the graph. It has a lower complexity of $\mathcal{O}(|V|^3)$, where 
$|V|$ is number of nodes in the graph, but is only applicable to bi-regular bipartite graphs. One drawback of this approach is that the recursive 
equations for calculating $N_k$ are tedious, particularly for values of $k \geq g+6$. The following theorem describes the general calculation of $N_i$, for any $g \leq i \leq 2g-2$, 
and the specifics of the calculation of $N_{g+4}$.

\begin{theo} \cite{dehghan2018spectrum}\label{T009}
	For a  given $(d_v, d_c)$-regular bipartite graph $G$, the number of $i$-cycles, $g \leq i \leq 2g-2$, is given by
	\begin{equation}
	N_i = [\sum_{j=1}^{|V|} \lambda_j^i - \Omega_{i}(d_v,d_c,G) - \Psi_{i}(d_v,d_c,G)]/(2i),
	\end{equation}
	where $\{\lambda_j\}_{j=1}^{|V|}$ is the spectrum of $G$, and $\Omega_{i}(d_v,d_c,G)$ and $\Psi_{i}(d_v,d_c,G)$ 
	are the number of closed cycle-free walks of length $i$ and closed walks with cycle of length $i$ in $G$, respectively. 
	For $i=g+4$, we have
	\begin{align*}
	\dfrac{\Psi_{g+4}(d_v,d_c,G)}{2(g+4)}&= N_{g+2}\times [\frac{g+2}{2}(d_v+d_c)-(g+2)]\\
	&+ N_g \times [\frac{g}{2}(d_v-2)(d_c-1)+\frac{g}{2}(d_c-2)(d_v-1)]\\
	&+ N_g \times \Big( [{{\frac{g}{2}} \choose 2}+\frac{g}{2}](d_v-2)^2+ [{{\frac{g}{2}} \choose 2}+\frac{g}{2}](d_c-2)^2 + (\frac{g}{2})^2(d_v-2)(d_c-2) \Big)\\
	&+ N_g \times \Big({g\choose 2}+2g + (g+2)\times (\frac{g}{2}(d_v-2)+\frac{g}{2}(d_c-2))\Big)\:,
	\end{align*}
	and
	\begin{equation}\label{E29999}
	\Omega_{g+4}(d_v,d_c,G) = n\times S_{d_v,d_c,g+4}+m\times S_{d_c,d_v,g+4}\:,
	\end{equation}
	where $n$ and $m$ are the number of variable and check nodes in $G$, respectively, and $S_{d_v,d_c,g+4}$ ($S_{d_c,d_v,g+4}$) represents the  number of  
	closed cycle-free walks of length $g+4$ from a variable node $v$ (a check node $c$) to itself. (Generating functions are used to compute functions 
	$S_{x,y,i}$ recursively~\cite{blake2017short}.)
\end{theo}

In this work, we investigate the relationship between the above two approaches. In particular, our goal is to find the relationship between the two spectra $\{\eta_i\}$ and 
$\{\lambda_i\}$ for bi-regular bipartite graphs. We show that the former spectrum includes eigenvalues $\pm 1$, $\pm \sqrt{-(d_v-1)}$, and $\pm \sqrt{-(d_c-1)}$. 
The remaining eigenvalues of $A_e$ are related to  the graph spectrum $\{\lambda_i\}$ through simple quadratic equations whose 
coefficients are determined by the node degrees $d_v$ and $d_c$. This allows one to compute $N_k$ using Theorem~\ref{T02}, but through the calculation of
the graph spectrum $\{\lambda_i\}$ rather than the direct calculation of $\{\eta_i\}$. As a result, the computational complexity reduces to $\mathcal{O}(|V|^3)$ rather than $\mathcal{O}(|E|^3)$, while avoiding the tedious equations of Theorem~\ref{T009}.

The organization of the rest of the paper is as follows: In Section~\ref{section02}, we present some definitions and notations.
Section~\ref{section04} contains our result on the relationship between the two spectra $\{\lambda_i\}$ and $\{\eta_i\}$,
and the derivation of the latter from the former. The paper is concluded in Section~\ref{section06}.

\section{Definitions and notations}
\label{section02}

A graph $G=(V,E)$ is a  set $V(G)$ of nodes and a multiset $E(G)$ of unordered pairs of nodes, called edges. If $\{v,u\}\in E$, we say that 
there is an edge between $v$ and $u$ (i.e., $v$ and $u$ are adjacent). We may also use notations $uv$ or $vu$ for the edge $\{v,u\}$. 
We say that a graph $G$ is simple, if it does not 
have any loop (i.e., no edge of the form $\{v,v\}$) or parallel edges (i.e.,  no two edges between the two same nodes).
A directed graph (digraph) $D=(V,E)$ is a  set $V$ of nodes and a  multiset $E$
of ordered pairs of nodes called arcs. For an arc $e = (u,w)$, we define the origin of $e$ to be $o(e) = u$, and the terminus of $e$ 
to be $t(e) = w$. The inverse arc of $e$, denoted by $\overline{e}$, is the arc formed by switching the
origin and terminus of $e$.  A digraph $D$ is called {\em symmetric} if whenever $(u,w)$ is an arc of $D$, its inverse arc $(w, u)$ is as well. 
For each graph $G$, its symmetric digraph $D(G)$ is defined by replacing each edge of $G$ with two arcs in opposite directions. See Fig. \ref{fig001}.
Thus, there is a simple correspondence between $G$ and $D(G)$.

\begin{figure}[ht]
	\begin{center}
		\includegraphics[scale=.5]{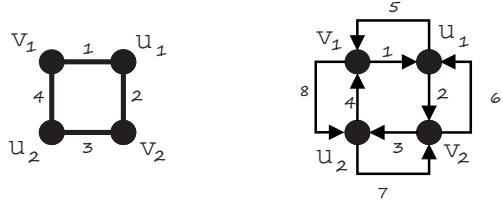}
		\caption{ A graph $G$ and its symmetric digraph $D(G)$.
		} \label{fig001}
	\end{center}
\end{figure}

In a graph $G$, the number of edges incident to a node $v$ is called the {\em degree} of $v$, and is denoted by $d(v)$.
Also, $\Delta(G)$ and $ \delta(G)$ are used to denote  the
maximum and minimum degree of $G$. For every node $v \in V (G)$, the set $N(v)$
denotes the set of neighbors of $v$ in $G$.

For a  graph $G$, a {\it walk} of length $c$ is a sequence of nodes $v_1, v_2, \ldots , v_{c+1}$ in $V$ such that $\{v_i, v_{i+1}\} \in E$, for all $i \in \{1, \ldots , c\}$. 
A walk can alternatively be represented by its sequence of edges.
A walk $v_1, v_2, \ldots , v_{k+1}$ is a {\it path} if all the nodes $v_1, v_2, \ldots , v_k$ are distinct. A walk is called a
{\it closed walk}  if the two end nodes are the same, i.e., if $v_1 = v_{k+1}$. Under the same condition, a path is called a {\it cycle}.
We denote cycles of length $k$, also referred to as $k$-cycles, by $C_k$. We use $N_k$ for $|C_k|$.  The length of the shortest cycle(s) 
in a graph is called {\em girth} and is denoted by $g$.

Consider a walk ${\cal W}$ of length $k$ represented by the sequence of edges $e_{i_1}, e_{i_2}, \ldots , e_{i_k}$. The
walk ${\cal W}$ is {\it backtrackless} if $e_{i_s}  \neq e_{i_{s+1}}$, for any $s \in\{1, \ldots , k-1\}$.
Also, the walk ${\cal W}$ is {\it tailless} if $e_{i_1}  \neq e_{i_{k}}$. In
this paper, we  use the term {\it TBC walk} to refer to a tailless backtrackless closed walk.  


A graph $G$ is {\it connected}, if there is a path between any two nodes of $G$. 
A graph $G=(V,E)$ is called {\it bipartite}, if the node set $V$ can be partitioned into two disjoint subsets $U$ and $W$, i.e., $V = U \cup W \text{ and } U \cap W =\emptyset $, such that every edge in $E$ connects a node from $U$ to a node from $W$. A graph is bipartite if and only if the lengths of all its cycles are even.
Tanner graphs of LDPC codes are bipartite graphs, in which $U$ and $W$ are referred to as {\it variable nodes} and {\it check nodes}, respectively. 
Parameters $n$ and $m$ in this case are used to denote $|U|$ and $|W|$, respectively. Parameter $n$ is the code's block length and the code rate 
$R$ satisfies $R \geq 1- (m/n)$. 

A bipartite graph $G = (U\cup W,E)$ is called {\it bi-regular}, if all the nodes on the same side of the bipartition have the same degree,
i.e., if all the nodes in $U$ have the same degree $d_u$ and all the nodes in $W$ have the same degree $d_w$. In the rest of the paper, we sometimes use 
notations $d_v$ and $d_c$ as a replacement for $d_u$ and $d_w$, respectively, to follow the notations commonly used in coding to denote
variable and check node degrees, respectively.
It is clear that, for a bi-regular graph, $|U|d_u=|W|d_w=|E(G)|$. A bipartite graph that is not bi-regular is called {\it irregular}.
A bipartite graph $G(U \cup W, E)$ is called {\em complete}, and is denoted by $K_{|U|,|W|}$, if every node in $U$ is connected to every node in $W$. 
The {\it degree sequences} of a bipartite graph $G$ are defined as the two monotonic non-increasing sequences of the node degrees on the two sides of the graph. 
For instance, the complete bipartite graph $K_{3,4}$ has degree sequences $(4,4,4)$ and $(3,3,3,3)$.

The {\it adjacency matrix} of a graph $G$ is a $|V| \times |V|$ matrix $A = [a_{ij}]$, where $a_{ij}$ is the number of edges connecting the node $i$ to the node
$j$, for all $i, j\in V$. Similarly, The   adjacency matrix  of a digraph $D$ is the matrix $A_D = [b_{ij}]$, where $b_{ij}$ is one if and only if 
$(i, j) \in E(D)$. The adjacency matrix $A$ is symmetric, and since we assumed that $G$ has no parallel edges, then $a_{ij}\in\{0, 1\}$, for all $i, j\in V$. Moreover, since $G$ has no loops, then $a_{ii} = 0$, for all $i \in V$.

An {\it eigenvalue} of $A$ is a number $\lambda $ such that $A\overrightarrow{v}=\lambda  \overrightarrow{v}$, for some nonzero vector $\overrightarrow{v}$. 
(Throughout the paper all vectors are assumed to be column vectors.)
The vector $\overrightarrow{v}$ is then called an {\it eigenvector} of $A$.
The set of the eigenvalues $\{\lambda_i\}$ of the adjacency matrix $A$ of a graph $G$ is called the {\em spectrum} of $G$. 
The determinant $\det(\lambda I- A)$, where $I$ is the identity matrix, is called the {\em characteristic polynomial} of $A$ (with variable $\lambda$).
The roots of this polynomial are the eigenvalues of $A$. An eigenvalue $ \lambda'$ of $A$ is said to have multiplicity $i$ if, when the characteristic polynomial is
factorized into linear factors, the factor $(\lambda - \lambda ')$ appears $i$ times. 
If $\lambda$ is an eigenvalue of $A$, then the subspace $ \{  \overrightarrow{v}: A\overrightarrow{v} = \lambda \overrightarrow{v} \}$ is 
called the {\em eigenspace} of $A$ associated with $\lambda$. The dimension of this eigensapce is at most the multiplicity of $\lambda$. 

There are some known results about the eigenvalues and eigenvectors of the adjacency matrix $A$ that we review below and use them in our work  (see, e.g.,~\cite{horn1990matrix}). (1) If $\lambda$ is an eigenvalue of $A$, then $\lambda^2 $ is an eigenvalue of $A^2$. (2) [Perron-Frobenius, Symmetric Case] Let $A$ be the adjacency 
matrix of a connected graph $G$, and let $ \lambda_1 \geq \lambda_2 \geq \ldots \geq \lambda_{|V|}$ be the spectrum of $G$. Then,  $ \lambda_1 > \lambda_2$ (i.e., 
the multiplicity of the largest eigenvalue of $A$ is one). (3) The largest eigenvalue of bi-regular bipartite graphs is $\sqrt{d_v d_c}$~\cite{MR2599858}. 
(4) A graph is bipartite if and only if its spectrum is symmetric about the origin.  (5) By Properties (2) and (4), in connected bipartite graphs,  
the multiplicity of the smallest eigenvalue is also one.
(6) By Property (4), for a given bipartite graph $G$, if  $\lambda_i$ is an eigenvalue of $A$ with multiplicity $m_i$, then $-\lambda_i$ is also an eigenvalue with multiplicity $m_i$.
Thus, the spectrum of $A$ has the following form $\{\pm \lambda_1^{m_1}, \ldots, \pm \lambda_r^{m_r}\}$, for some $r \geq 1$, and we have $\sum_{i=1}^r 2\times m_i=|V|$.
(7) The adjacency matrix $A$ of $G$ has $|V(G)|$ linearly independent eigenvectors, such that for each $1 \leq i \leq r$, there
are $m_i$ linearly independent eigenvectors associated with each eigenvalue $\lambda_i$ and $-\lambda_i$.


Another important property of the adjacency matrix is that the number of walks
between any two nodes of the graph can be determined using the powers of this matrix. In other words, the entry in
the $i^{\text{th}}$ row and the $j^{\text{th}}$ column of $A^k$, $[A^k]_{ij}$ , is the number of walks of length $k$ between nodes $i$ and $j$. Consequently, the total number of closed walks of length $k$ in $G$ is  $tr(A^k)$, where $tr(\cdot)$ is the trace of a matrix. 
It is well-known that $tr(A^k)= \sum_{i=1}^{|V|}\lambda_i^k$, and thus the multiplicity of closed walks of different length 
in a graph can be obtained using the spectrum of the graph.

For a given graph $G$, the directed edge matrix $A_e$, is a $2|E|\times 2|E|$ matrix defined as follows. For each edge $e_i=\{v,u\}$ in $G$, 
we consider two opposite arcs $(v,u),(u,v)$, and denote them by $f_i$ and $f_{|E(G)|+i}$ (i.e., $ f_i = \overline{f_{|E(G)|+i}}$). We then define
 
  \begin{equation}\label{EEEE1}
 	 (A_e)_{i,j}=
 	\begin{cases}
 	1,      &\text{if }\,\,t(f_i)=o(f_j) \text{ and } f_i \neq \overline{f_j}\\
 	0,       & \text{otherwise}.\
 	\end{cases} 
\end{equation}
 
In other words, for a given graph $G$, we consider its associated symmetric digraph $D(G)$, and then calculate $A_e$ from $D(G)$ using (\ref{EEEE1}). 
For example, for graphs $G$ and $D(G)$ in Fig. \ref{fig001}, we have

\begin{center}
	$A_e=    
	\left[
	\begin{array}{cccc|cccc}
	0 & 1 & 0 & 0 & 0 & 0 & 0  & 0 \\
	0 & 0 & 1 & 0 & 0 & 0 & 0  & 0 \\
	0 & 0 & 0 & 1 & 0 & 0 & 0  & 0 \\
	1 & 0 & 0 & 0 & 0 & 0 & 0  & 0 \\
	\hline
	0 & 0 & 0 & 0 & 0 & 0 & 0  & 1 \\
	0 & 0 & 0 & 0 & 1 & 0 & 0  & 0 \\
	0 & 0 & 0 & 0 & 0 & 1 & 0  & 0 \\
	0 & 0 & 0 & 0 & 0 & 0 & 1  & 0 \\	
	\end{array}
	\right]
	$\:.
\end{center}

The number of $k$-cycles, $g \leq k \leq 2g-2$, in a bipartite graph $G$ can be obtained from the spectrum $\{\eta_i\}$ of $A_e$ using Theorem~\ref{T02}.



The rank of a matrix $B$, denoted by $Rank(B)$, is the dimension of the vector space generated  by its columns. This corresponds to the maximum number of linearly independent columns of $A$. The rank is also the dimension of the space spanned by the rows of $B$. Thus, if $B$ is an $m\times n$ matrix, then
\begin{equation}\label{EE99}
Rank(B)=Rank(B^t) \leq \min\{m,n\}\:,
\end{equation}
where $B^t$  is the transpose of $B$. The kernel (null space) of a matrix $B$ is the set of solutions to the equation $B \overrightarrow{x} = \overrightarrow{0}$, 
where $\overrightarrow{0}$ is the zero vector. The dimension of the null space of $B$ is called the nullity of $B$ and is denoted by $Null(B)$. 
For an $m \times n$ matrix $B$, we have (Rank-Nullity Theorem):
\begin{equation}
Rank(B)+Null(B)=n\:.
\end{equation}

\section{The Relationship between the Spectra of $A_e$ and $A$ for Bi-regular Bipartite Graphs, and the New Method to Count Short Cycles}
\label{section04}


In \cite{MR2794074}, it was shown that for a regular graph $G$, the eigenvalues of $A_e$ can be computed from those of $A$.
A key component in the derivations of \cite{MR2794074} is the special properties that $A_e$ has as a result of the regularity of the graph. 
For the bi-regular graphs, considered in this work, however, such properties do not exist and thus the derivations are much different.
In this section, we derive the spectrum $\{\eta_i\}$ of $A_e$ from the graph spectrum $\{\lambda_i\}$ for bi-regular bipartite graphs,
and then use the results to count the short cycles of the graph by Theorem~\ref{T02}. 

To derive our results, we first define an auxiliary matrix $\widetilde{A}$ as a function of $A$. We then find the eigenvalues $\{\xi_i\}$ of $\widetilde{A}^2$, which are on the one hand 
related to  $\{\lambda_i\}$, and on the other hand to $\{\eta_i\}$. Through these relationships, we derive $\{\eta_i\}$ from $\{\lambda_i\}$. In the following, for simplicity, we use notations $q_1$ and $q_2$ to denote $d_v-1$ and $d_c-1$, respectively.

For a bi-regular bipartite graph $G=(U \cup W, E)$, let $ \widetilde{A}=[\widetilde{a}_{(u,w),(x,y)}]_{u,w,x,y\in V(G)}$ be a $|V(G)|^2\times |V(G)|^2$ 
matrix such that the entries of $\widetilde{A}$ are given by
\begin{equation}\label{E001}
\widetilde{a}_{(u,w),(x,y)} =a_{uw}a_{xy} \delta_{wx}(1- \delta_{ uy})\:,
\end{equation}
where $\delta_{uw}$ is the Kronecker delta (which is equal to $1$ if $u=w$, and equal to zero, otherwise), and $a_{uw}$ is the $(u,w)^{th}$ entry of the adjacency matrix $A$ of $G$.
In the rest of the paper, we assume that the rows and columns of $\widetilde{A} $ are sorted in the following order: 
First, the set $\{(u,w): u\in U, w\in W, uw\in E(G)\}$, second  $\{(w,u): u\in U, w\in W, uw\in E(G)\}$, and finally, 
other pairs $\{(u,w), (w,u): u\in U, w\in W, uw \notin E(G)\}$. 
Note that the union of the first two sets is the set of directed edges in the symmetric digraph $D(G)$ associated with $G$. 
Also, by (\ref{E001}), $\widetilde{a}_{(u,w),(x,y)}=1$ if and only if we have
\begin{itemize}
\item [(i)] $a_{uw}a_{xy}=1$ (i.e., $f_i=(u,w), f_j=(x,y) \in  E(D(G))$),
\item [(ii)] $\delta_{wx}=1$ (i.e., $t(f_i) = o(f_j)$), and 
\item [(iii)] $(1- \delta_{ uy})=1$ (i.e., $f_i \neq \overline{f_j})$).
\end{itemize}
Thus, by  (\ref{EEEE1}), 
 the matrix $\widetilde{A}$ has the following form
\begin{equation}\label{EEE1}
\widetilde{A}=\left[
\begin{array}{c|l}
A_e & 0_{(2|E|)\times(|V|^2-2|E|)} \\
\hline
0_{(|V|^2-2|E|)\times(2|E|)} & 0_{(|V|^2-2|E|)\times(|V|^2-2|E|)}   \\	
\end{array}
\right]\:,
\end{equation}
and by (\ref{EEE1}), we have the following result.

\begin{lem}
The eigenvalues of $\widetilde{A}$ are the same as those of $A_e$ with the addition of $|V|^2-2|E| $ zero eigenvalues.
\label{lema1}
\end{lem}

Furthermore, since $G$ is bipartite, and based on the labeling of rows and columns (i.e.,  first, are listed pairs $\{(u,w): u\in U, w\in W, uw\in E(G)\}$, followed by pairs 
$\{(w,u): u\in U, w\in W, uw\in E(G)\}$), $A_e$ has the following form
\begin{equation}\label{EEE78}
A_e=\left[
\begin{array}{l|l}
0_{|E|\times|E|} & B_e \\
\hline
C_e             & 0_{|E|\times|E|}   \\	
\end{array}
\right]\:,
\end{equation}
where $B_e$  and $C_e$ are $ |E|\times  |E|$ matrices. As an example, by the ordering just described ($(u_1,v_1), (u_1,v_2), (u_2,v_2), (u_2,v_1)$ are the first $4$ arcs, followed by their inverse arcs in the same order), for the graph $G$ shown in Fig. \ref{fig001}, we have

\begin{center}
	$A_e=    
	\left[
	\begin{array}{cccc|cccc}
	0 & 0 & 0 & 0 & 0 & 0 & 0  & 1 \\
	0 & 0 & 0 & 0 & 0 & 0 & 1  & 0 \\
	0 & 0 & 0 & 0 & 0 & 1 & 0  & 0 \\
	0 & 0 & 0 & 0 & 1 & 0 & 0  & 0 \\
	\hline
	0 & 1 & 0 & 0 & 0 & 0 & 0  & 0 \\
	1 & 0 & 0 & 0 & 0 & 0 & 0  & 0 \\
	0 & 0 & 0 & 1 & 0 & 0 & 0  & 0 \\
	0 & 0 & 1 & 0 & 0 & 0 & 0  & 0 \\	
	\end{array}
	\right]\:.
	$
\end{center}

From (\ref{EEE1}) and (\ref{EEE78}), one can see that the matrix $\widetilde{A}^2$ has the following form:
\begin{equation}\label{E003}
\widetilde{A}^2=\left[
\begin{array}{c|l}
\begin{array}{l|l}
B_e C_e & 0_{|E|\times|E|}  \\
\hline
0_{|E|\times|E|}            & C_e B_e  \\	
\end{array} & 0_{(2|E|)\times(|V|^2-2|E|)} \\
\hline
0_{(|V|^2-2|E|)\times(2|E|)} & 0_{(|V|^2-2|E|)\times(|V|^2-2|E|)}   \\	
\end{array}
\right]\:,
\end{equation}
or equivalently,
\begin{equation}\label{E0033}
\widetilde{A}^2=\left[
\begin{array}{c|l}
A_e^2 & 0_{(2|E|)\times(|V|^2-2|E|)} \\
\hline
0_{(|V|^2-2|E|)\times(2|E|)} & 0_{(|V|^2-2|E|)\times(|V|^2-2|E|)}   \\	
\end{array}
\right]
\end{equation}
It is easy to see that $((u,w),(x,y))^{th}$ entry of the  element of $\widetilde{A}^2 $ (denoted by $\widetilde{a}^2_{(u,w),(x,y)}$) is given by
\begin{equation}\label{E006}
\widetilde{a}^2_{(u,w),(x,y)}=
\begin{cases}
1,       & \text{if}\,\,uw,wx,xy\in E,\,\, x\neq u, \,\, y\neq w\\
0,       & \text{otherwise}.\
\end{cases}
\end{equation}
We thus have 
\begin{equation}\label{E007}
\widetilde{a}^2_{(u,w),(x,y)}=
a_{uw}a_{wx}a_{xy}(1-\delta_{xu})(1-\delta_{yw})\:.
\end{equation}

Next, we study the structure of eigenvectors of $\widetilde{A}^2$.

\begin{lem}
    Consider a number $\xi \neq 0$ and a vector $\overrightarrow{\phi}$ of size $|V|^2$, and 
	denote the element that corresponds to the pair $(x,y)$ in the vector $\overrightarrow{\phi}$ by $\phi_{(x,y)}$.
	Then,  $\overrightarrow{\phi}$ is an eigenvector of $\widetilde{A}^2$ associated with eigenvalue $\xi$ if and only if,
	for each pair $(u,w)$, where $u\in U$ and $w\in W$, we have
	\begin{equation}\label{E004}
	\xi \phi_{(u,w)}=a_{uw}\sum_{x\in U} a_{wx}\sum_{y\in W} a_{xy} \phi_{(x,y)}
	-a_{uw}\sum_{x\in U} a_{wx}  \phi_{(x,w)}
	-a_{uw}\sum_{y\in W} a_{uy}   \phi_{(u,y)}
	+a_{uw}\phi_{(u,w)}\:,
	\end{equation}
	and for each pair $(w,u)$, where  $w\in W$ and $u\in U$, we have
	\begin{equation}\label{EE001}
	\xi \phi_{(w,u)}=a_{wu}\sum_{y\in W} a_{uy}\sum_{x\in U} a_{yx} \phi_{( y,x)}
	-a_{wu}\sum_{y\in W} a_{uy}  \phi_{(y,u)}
	-a_{wu}\sum_{x\in U} a_{wx}   \phi_{(w,x)}
	+a_{wu}\phi_{(w,u)} \:,
	\end{equation}
	and for all the other pairs $(x,y)$, $\phi_{(x,y)}=0$.
	\label{lemc}
\end{lem}

\begin{proof}
By the definition of eigenvalue/eigenvector and (\ref{E007}), it is clear that for $\xi \neq 0$, we must have $\phi_{(x,y)}=0$, for all cases 
where nodes $x$ and $y$ are on the same side of the graph. On the other hand, for each pair $(u,w)$, where $u\in U$ and $w\in W$, 
by the definition of eigenvalue/eigenvector and (\ref{E007}), we have:
	\begin{align*}
	\xi \phi_{(u,w)}&=\sum_{x\in U}\sum_{y\in W} a_{uw}a_{wx}a_{xy}(1-\delta_{xu})(1-\delta_{yw})\phi_{(x,y)}\\
	&=a_{uw}\sum_{x\in U}a_{wx}\sum_{y\in W} a_{xy}(1-\delta_{xu})(1-\delta_{yw})\phi_{(x,y)}\\
	&=a_{uw}\sum_{x\in U} a_{wx}\sum_{y\in W} a_{xy} \phi_{(x,y)}
	-a_{uw}\sum_{x\in U} a_{wx}  \phi_{(x,w)}
	-a_{uw}\sum_{y\in W} a_{uy}   \phi_{(u,y)}
	+a_{uw}\phi_{(u,w)}
	\end{align*} 
	Equation (\ref{EE001}) is derived similarly.
\end{proof}

\subsection{From the non-zero eigenvalues of $A$ to the eigenvalues of $\widetilde{A}^2$}

\begin{lem}
Let  $\lambda \neq 0$ be an eigenvalue of the adjacency matrix $A$. Then the solutions of the quadratic equation 
$\xi^2 +(-\lambda^2+q_1+q_2)\xi+q_1q_2=0$ are two eigenvalues of  $\widetilde{A}^2$.
\end{lem}

\begin{proof}
 Let $\lambda$ be an eigenvalue of the adjacency matrix $A$ with a corresponding eigenvector $ \overrightarrow{\mu}=[\mu_{u_1},\ldots, \mu_{u_{n}}, \mu_{w_1},\ldots, \mu_{w_{m}}]^t$ (note that the elements of the eigenvector are sorted by listing the elements corresponding to the nodes in $U$ first, followed by those corresponding to 
 the nodes in $W$). 
 By using $ \overrightarrow{\mu}$, we define a vector $\overrightarrow{\phi}$ of size $|V|^2$ in the following way (the element corresponding to the pair $(x,y), x \in V, y \in V$, in 
 $\overrightarrow{\phi}$ is denoted by $\phi_{(x,y)}$):
\begin{equation}\label{E008}
\phi_{(x,y)}=
\begin{cases}
a_{xy}(\mu_y- f_1\mu_x),       & \text{if}\,\,x\in U,\,\,y\in W,\\
a_{xy}(\mu_y- f_2\mu_x),       & \text{if}\,\,x\in W,\,\,y\in U,\\
0,       & \text{otherwise},\
\end{cases}
\end{equation}
where $f_1$ and $f_2$ are constant numbers. 
Now, we show that by the proper choice of $f_1$ and $f_2$, the vector $\overrightarrow{\phi}$ is an eigenvector of 
$\widetilde{A}^2$, and in the process find the corresponding eigenvalues $\xi$.

By substituting (\ref{E008}) in  (\ref{E004}), we have:
\begin{alignat*}{2}\label{E009}
\xi \phi_{(u,w)}
&=  &&a_{uw}\sum_{x\in U} a_{wx}\sum_{y\in W} a_{xy} (\mu_y- f_1\mu_x)\\
&   &&-a_{uw}\sum_{x\in U} a_{wx} (\mu_w- f_1\mu_x)\\
&   &&-a_{uw}\sum_{y\in W} a_{uy}   (\mu_y- f_1\mu_u)\\
&   &&+a_{uw}(\mu_w- f_1\mu_u)\\
&=  &&a_{uw}\sum_{x\in U} a_{wx}\Big( \lambda \mu_x-(q_1+1)f_1 \mu_x \Big)\\
&   &&-a_{uw}\mu_w(q_2+1)+a_{uw}f_1 \lambda \mu_w\\
&   &&-a_{uw}\lambda \mu_u + a_{uw} f_1 \mu_u (q_1+1)\\
&   &&+a_{uw}\mu_w - a_{uw}f_1 \mu_u\\
&=  &&a_{uw}\lambda^2 \mu_w-a_{uw}(q_1+1)f_1\lambda \mu_w\\
&   &&-a_{uw}\mu_w(q_2+1)+a_{uw}f_1 \lambda \mu_w\\
&   &&-a_{uw}\lambda \mu_u + a_{uw} f_1 \mu_u (q_1+1)\\
&   &&+a_{uw}\mu_w - a_{uw}f_1 \mu_u\\
&=  && a_{uw}\mu_w\Big(\lambda^2- \lambda f_1 q_1-q_2 \Big)\\
&   &&- a_{uw}\mu_u \Big( \lambda - f_1 q_1\Big)\:, \numberthis
\end{alignat*} 
where in the second and third last steps, we have used the definition of eigenvalue/eigenvector of $A$.
From (\ref{E009}), and considering $\xi \neq 0$, we have:
\begin{equation}\label{E010}
\xi \phi_{(u,w)} = 	a_{uw} \xi \Big(\dfrac{\lambda^2- \lambda f_1 q_1-q_2}{\xi}\mu_w - \dfrac{\lambda -f_1 q_1}{\xi}\mu_u \Big)\:.
\end{equation}
From (\ref{E010}) and (\ref{E008}), we obtain:
\begin{equation}\label{E011}
\begin{cases}
\dfrac{\lambda^2- \lambda f_1 q_1-q_2}{\xi}=1\\
\dfrac{\lambda - f_1
	q_1}{\xi}=f_1
\end{cases}
\end{equation}
By solving (\ref{E011}), we have (note that since $\lambda\neq 0$, by (\ref{E011}), we have $ \xi \neq -q_1$):
\begin{equation}\label{E012}
f_1=\dfrac{\lambda}{\xi+q_1}\:,
\end{equation}
and
\begin{equation}\label{E013}
\xi^2 +(-\lambda^2+q_1+q_2)\xi+q_1q_2=0\:.
\end{equation}

Similarly, by substituting (\ref{E008}) in  (\ref{EE001}), and taking the same steps as those taken in the derivation of~(\ref{E009}), we have:
\begin{equation}\label{EE004}
\xi \phi_{(w,u)} = 	a_{wu} \xi \Big(\dfrac{\lambda^2- \lambda f_2 q_2-q_1}{\xi}\mu_u - \dfrac{\lambda - f_2 q_2}{\xi}\mu_w \Big)\:.
\end{equation}

From  (\ref{EE004}) and (\ref{E008}), we have:
\begin{equation}\label{EE003}
\begin{cases}
\dfrac{\lambda^2- \lambda f_2 q_2-q_1}{\xi}=1\\
\dfrac{\lambda-f_2 q_2}{\xi}=f_2\:.
\end{cases}
\end{equation}

By solving (\ref{EE003}), we obtain (since $\lambda \neq 0$, by (\ref{EE003}), $ \xi \neq -q_2$):
\begin{equation}\label{E005}
f_2=\dfrac{\lambda}{\xi+q_2}\:,
\end{equation}
and the same equation as in (\ref{E013}).

Therefore, by solving (\ref{E013}), we find the eigenvalues $\xi$ of $\widetilde{A}^2$ corresponding to $\lambda$, and then by
substituting the obtained $\xi$ in (\ref{E012}) and (\ref{E005}), we find the constants $f_1$ and $f_2$. These are then replaced in (\ref{E008}) to obtain
the corresponding eigenvectors of  $\widetilde{A}^2$.
\end{proof}

Next, we discuss how the eigenvalues of $A_e$ can be computed from those of $\widetilde{A}^2$.

\subsection{From the spectrum of $\widetilde{A}^2$ to that of $A_e$}

\begin{lem}\label{lem01}  \cite{karimi2012counting}
Let $G$ be a bi-regular bipartite graph and $A_e$ be its directed edge matrix. Then, the
eigenvalues of $A_e$ are symmetric with respect to the origin. Moreover, $\eta^2$ is an eigenvalue
of $A_e^2$ if and only if $ \pm \eta$ are eigenvalues of $A_e$.
\end{lem}

\begin{lem}\label{lem02}
Let $G$ be a bi-regular bipartite graph. Then the spectrum of $\widetilde{A}$ can be computed from that of $\widetilde{A}^2$, i.e., if $\widetilde{A}^2$ has an eigenvalue $\xi$ with multiplicity $m$, then $\widetilde{A}$ has eigenvalues $ \pm \sqrt{\xi}$, each with multiplicity $m/2$. 	 
\end{lem}

\begin{proof}
The proof follows from Lemma \ref{lem01},  (\ref{EEE1}) and (\ref{E0033}).
\end{proof}
Using Lemmas~\ref{lema1} and \ref{lem02}, one can obtain the spectrum of $A_e$ from that of $\widetilde{A}^2$.

\subsection{From the spectrum of $A$ to that of $A_e$}

\begin{theo}\label{T03}  
Let $G=(V=U \cup W,E)$ be a connected bi-regular bipartite graph such that each node in $U$ has degree $q_1+1$ and each node in $W$ has degree $q_2+1$, where $q_2 \geq 2$, $q_1 \geq 1$ and $q_2 \geq q_1$. Also, assume that $|U|=n$ and $|W|=m$. 	
The eigenvalues of the directed edge matrix $A_e$ of $G$ can then be computed from the eigenvalues of the adjacency matrix $A$ as follows:\\
{\bf \underline{Step 1.}}  For each strictly negative eigenvalue 
$\lambda$ of $A$, use Equation (\ref{E013}) to find two solutions. For each solution $\xi \neq 1$, the numbers  $\pm\sqrt{\xi}$ are eigenvalues of $A_e$, each with the same multiplicity as that of $\lambda$ in the spectrum of $A$.
%
(The total number of eigenvalues of $A_e$ obtained in this step is $2(m+n)-2Null(A)-2$.)\\
{\bf\underline{Step 2.}} Matrix $A_e$ also has the eigenvalues $\pm\sqrt{-q_1}$ and $\pm\sqrt{-q_2}$. The multiplicity of each of 
the eigenvalues $\pm\sqrt{-q_1}$ ($\pm\sqrt{-q_2}$) is $ n-Rank(A)/2$ ($ m-Rank(A)/2 $).
\footnote{Note that $\pm\sqrt{-q_1}$ and $\pm\sqrt{-q_2}$ are solutions of (\ref{E013}) for $\lambda=0$.} (The total number of  of eigenvalues
of $A_e$ obtained in this step is $2(m+n)- 2Rank(A)=2Null(A)$.)\\
{\bf\underline{Step 3.}} Furthermore, Matrix $A_e$ has eigenvalues $\pm 1$, each with multiplicity 
 $|E|-(m+n)+1$. (The total number of eigenvalues in this step is $2|E|-2(m+n)+2$.)
\end{theo}

\begin{proof}
In the following, we find the set of eigenvalues of $\widetilde{A}^2$ and their multiplicities, and then use 
Lemmas \ref{lema1} and \ref{lem02} to obtain the set of eigenvalues of $A_e$. 

Suppose that the spectrum of $A$ is $\{\pm \lambda_1^{m_1}, \ldots, \pm \lambda_r^{m_r}\}$, for some $r \geq 1$, where $\sum_{i=1}^r 2\times m_i=|V|$.
For each $i$, $1 \leq i \leq r$, there are $m_i$ linearly independent eigenvectors $\overrightarrow{\mu}_{i,1}, \ldots, \overrightarrow{\mu}_{i,m_i}$, 
associated with the eigenvalue $\lambda_{i}$. 

For each $i$, let $\xi_{i_1}$ and $\xi_{i_2}$ be the two eigenvalues obtained from (\ref{E013}) by replacing $\lambda$ by $\lambda_{i}$ (note that the solutions of (\ref{E013}) for $\lambda=-\lambda_i$ are the same as those for $\lambda=\lambda_{i}$). 
We consider three cases that cover all possible scenarios. Case A: $\lambda_{i}\neq 0$ and $\xi_{i_1}\neq 1 $; Case B: $\lambda_{i}= 0$; and 
Case C: $\lambda_{i}\neq 0$ and $ \xi_{i_1}=1$. (Cases A, B and C correspond to Steps 1, 2 and 3 of the derivation of all the eigenvalues of $A_e$. 
Note that, for each of Cases A, B and C, in the following, we find a lower bound on the multiplicity of the eigenvalues of $A_e$ (or those of $\widetilde{A}^2$) that are obtained in those cases. Based on the fact that the sum of the obtained lower bounds is equal to $2|E|$ ($|V|^2$) for $A_e$ ($\widetilde{A}^2$), we conclude that in each case, the multiplicity of the eigenvalues is exactly equal to the lower bound.)
\\
{\bf Case A. ($\lambda_{i}\neq 0$ and $\xi_{i_1}\neq 1$)}
In this case, we show that for each $i$, the multiplicity of $\xi_{i_1}$ is at least $2\times m_i$.\footnote{As explained before, this lower bound is tight.}  

Consider vectors $\overrightarrow{\phi}_{i,1}, \ldots,\overrightarrow{\phi}_{i,m_i}$, each of size $|V|^2$, corresponding to 
eigenvectors $\overrightarrow{\mu}_{i,1}, \ldots, \overrightarrow{\mu}_{i,m_i}$ of $A$ associated with eigenvalue $\lambda_i$, respectively.
Assume that the element $(x,y), x \in V, y \in V$, of each vector $\overrightarrow{\phi}_{i,j}$ is derived from the elements of the corresponding vector
$\overrightarrow{\mu}_{i,j}$ using the following equation:
 \begin{equation}\label{EEE02}
\phi_{(x,y)}=
\begin{cases}
a_{xy}(\mu_y- f_1\mu_x),       & \text{if}\,\,x\in U,\,\,y\in W,\\
0,       & \text{otherwise},\
\end{cases}
\end{equation}
where $f_1=\dfrac{\lambda_{i}}{\xi_{i_1}+q_1}$. Using simple calculations, one can see that 
for each $j$, we have $\widetilde{A}^2 \overrightarrow{\phi}_{i,j} = \xi_{i_1} \overrightarrow{\phi}_{i,j}$, and thus, $\overrightarrow{\phi}_{i,1}, \ldots,\overrightarrow{\phi}_{i,m_i}$ are eigenvectors associated with the eigenvalue  $\xi_{i_1}$.


Also, consider vectors $\overrightarrow{\rho}_{i,1}, \ldots,\overrightarrow{\rho}_{i,m_i}$, each of size $|V|^2$, corresponding to 
eigenvectors $\overrightarrow{\mu}_{i,1}, \ldots, \overrightarrow{\mu}_{i,m_i}$ of $A$ associated with eigenvalue $\lambda_i$, respectively.
Assume that the element $(x,y), x \in V, y \in V$, of each vector $\overrightarrow{\rho}_{i,j}$ is derived from the elements of the corresponding vector
$\overrightarrow{\mu}_{i,j}$ using the following equation:
\begin{equation}\label{EEE03}
\rho_{(x,y)}=
\begin{cases}
a_{xy}(\mu_y- f_2\mu_x),       & \text{if}\,\,x\in W,\,\,y\in U,\\
0,       & \text{otherwise},\
\end{cases}
\end{equation}
where $f_2=\dfrac{\lambda_{i}}{\xi_{i_1}+q_2}$. For each $j$, we have $\widetilde{A}^2 \overrightarrow{\rho}_{i,j} = \xi_{i_1} \overrightarrow{\rho}_{i,j}$, and thus, vectors $\overrightarrow{\rho}_{i,1}, \ldots,\overrightarrow{\rho}_{i,m_i}$ are also eigenvectors associated with the eigenvalue  $\xi_{i_1}$.

Regarding the dependency within each of the two groups of eigenvectors $\{\overrightarrow{\phi}_{i,j}\}$ and $\{\overrightarrow{\rho}_{i,j}\}$, we have 
the following fact whose proof is provided in Appendix \ref{A2}.

{\bf Fact 1.} The vectors $\overrightarrow{\phi}_{i,1}, \ldots,\overrightarrow{\phi}_{i,m_i}$ are linearly independent. So are the vectors 
$\overrightarrow{\rho}_{i,1}, \ldots,\overrightarrow{\rho}_{i,m_i}$.

Fact 1 together with the fact that there is no overlap between the location of non-zero elements  in any vector in the set $\{\overrightarrow{\phi}_{i,j}\}$ and that of 
any vector in the set $\{\overrightarrow{\rho}_{i,j}\}$ prove that the multiplicity of $\xi_{i_1}$, in Case A, is at least $2\times m_i$.

By Lemmas~\ref{lema1} and \ref{lem02}, the number of eigenvalues $\eta$ of $A_e$ that are obtained from Case A is the same as the number of eigenvalues 
$\xi$ of $\widetilde{A}^2$ that are obtained for this case. To count the total number of eigenvalues $\xi$ of $\widetilde{A}^2$, we note that 
the total number of non-zero eigenvalues $\lambda$ of $A$ is $m+n-Null(A)$, out of which half are negative. This together with the fact that each eigenvalue $\lambda$
results in two eigenvalues $\xi$ and that if the multiplicity of $\lambda$ is $m$, then the multiplicity of each resulting $\xi$ is $2m$ implies that the total
number of eigenvalues $\xi$ is $2(m+n-Null(A))$. For Case A, however, we have excluded $\xi = 1$. It is easy to see that (\ref{E013}) has a solution $\xi=1$ if and only if
$\lambda =  \pm\sqrt{(1+q_1)(1+q_2)}$. (The other solution of (\ref{E013}) in this case is $\xi=q_1 q_2$.) 
These are the two eigenvalues of $A$ with the largest magnitude (and each with multiplicity one). Excluding $\xi = 1$, which has multiplicity two, means that for $\lambda = - \sqrt{(1+q_1)(1+q_2)}$, rather than four $\xi$ values, we only have two counted in Case A ($\xi =  q_1 q_2$ with multiplicity two). This reduces the total number
of eigenvalues $\xi$ for Case A to $2(m+n-Null(A)) - 2$.



{\bf Case B. ($\lambda_{i}= 0$)}
For this case, in the following, we show that we have two eigenvalues $\xi_{i_1}=-q_1$ and $\xi_{i_2}=-q_2$ for $\widetilde{A}^2$. 
(Note that these eigenvalues are in fact the solutions of (\ref{E013}) for $\lambda_{i}= 0$.) These eigenvalues, 
based on Lemmas~\ref{lema1} and \ref{lem02}, result in eigenvalues $\pm \sqrt{-q_1}$ and $\pm \sqrt{-q_2}$ for $A_e$. 
In the following, we also prove that the multiplicities of the eigenvalues $\xi_{i_1}$ and $\xi_{i_2}$ of $\widetilde{A}^2$ are $2n- Rank(A)$ and $2m- Rank(A)$, respectively. This together with Lemma~\ref{lem02} prove the claim of the theorem for the multiplicities of 
eigenvalues $\pm \sqrt{-q_1}$ and $\pm \sqrt{-q_2}$ of $A_e$.

To prove that $\xi_{i_1}=-q_1$ and $\xi_{i_2}=-q_2$ are eigenvalues of $\widetilde{A}^2$, and to obtain their multiplicities, we note that 
the graph $G$ is bipartite, and thus its adjacency matrix has the following form
\[
A =\left[
\begin{array}{l|l}
0_{ n\times n}  & D_{n\times m} \\
\hline
D_{m\times n}^t          & 0_{ m\times m}    \\	
\end{array}
\right]\:.
\]
As a result, we have the following fact whose proof is presented in Appendix \ref{A3}.

{\bf Fact 2.} We have
\begin{equation}\label{EEE54}
Null(D)=m-Rank(A)/2 \:,
\end{equation}
and
\begin{equation}\label{EEE55}
Null(D^t)=n-Rank(A)/2\:.
\end{equation}

Let $\overrightarrow{\mu}_{1}, \ldots, \overrightarrow{\mu}_{t}$,  where $t = m-Rank(A)/2$, be the linearly independent eigenvectors of matrix $D$ associated with
eigenvalue $0$. Corresponding to each vector $\overrightarrow{\mu}_{i}$ in the null space of $D$, we define the following two vectors  $\overrightarrow{\phi}_{i}$ and 
$\overrightarrow{\phi}_{i}'$, each of size $|V|^2$:
\begin{equation}\label{EEE56}
\phi_{(x,y)}=
\begin{cases}
a_{xy}\mu_y,       & \text{if}\,\,x\in U,\,\,y\in W, \\
a_{xy}\mu_x,       & \text{if}\,\,x\in W,\,\,y\in U, \\
0,                         & \text{otherwise},\
\end{cases}
\end{equation}
and
\begin{equation}\label{EEE57}
\phi_{(x,y)}'=
\begin{cases}
a_{xy}\mu_y,       & \text{if}\,\,x\in U,\,\,y\in W, \\
0,       & \text{otherwise},\ 
\end{cases}
\end{equation}
where $\phi_{(x,y)}$ ($\phi_{(x,y)}'$) is the element of $\overrightarrow{\phi}_{i}$ ($\overrightarrow{\phi}_{i}'$) corresponding to the pair of nodes $(x,y)$,
and $\mu_x$ ($\mu_y$) is the element of $\overrightarrow{\mu}_{i}$ corresponding to node $x$ ($y$) $\in W$. We then have the following result whose 
proof is provided in Appendix \ref{A4}.

{\bf Fact 3.} Vectors $\overrightarrow{\phi}_{i}$  and $\overrightarrow{\phi}_{i}'$ are eigenvectors of $\widetilde{A}^2$ associated with eigenvalue $-q_2$. 

Since the vectors $\overrightarrow{\mu}_{1}, \ldots, \overrightarrow{\mu}_{t} $ are linearly independent, then by the definitions (\ref{EEE56})  and (\ref{EEE57}), 
the vectors $\overrightarrow{\phi}_{1},\overrightarrow{\phi}_{1}',\ldots,\overrightarrow{\phi}_{t},\overrightarrow{\phi}_{t}' $ are also linearly independent. 
This implies that the multiplicity of the eigenvalue $-q_2$ of $\widetilde{A}^2$ is at least $2t = 2m-Rank(A)$.

Similarly, corresponding to each vector $\overrightarrow{\mu}_{i}, 1 \leq i \leq n-Rank(A)/2$, in the null space of $D^t$, we define the following 
two vectors  $\overrightarrow{\phi}_{i}$ and  $\overrightarrow{\phi}_{i}'$:
\begin{equation}\label{EEE59}
\phi_{(x,y)}=
\begin{cases}
a_{xy}\mu_x,       & \text{if}\,\,x\in U,\,\,y\in W, \\
a_{xy}\mu_y,       & \text{if}\,\,x\in W,\,\,y\in U, \\
0,                         & \text{otherwise}\,
\end{cases}
\end{equation}
and
\begin{equation}\label{EEE60}
\phi_{(x,y)}'=
\begin{cases}
a_{xy}\mu_x,       & \text{if}\,\,x\in U,\,\,y\in W, \\
0,       & \text{otherwise}.\ 
\end{cases}
\end{equation}
Similar to the proof of Fact 3, it can be seen that these $2n-Rank(A)$ vectors are eigenvectors of $\widetilde{A}^2$ associated with eigenvalue $-q_1$. Moreover, they are linearly 
independent, and thus, the multiplicity of $-q_1$ is at least $2n-Rank(A)$.\footnote{Note that, based on the total multiplicity of the eigenvalues of $\widetilde{A}^2$, the multiplicity of the eigenvalues $-q_2$ and $-q_1$ of $\widetilde{A}^2$ is equal to $2m-Rank(A)$ and $2n-Rank(A)$, respectively.}  

Finally, the sum of multiplicities of the eigenvalues $-q_1$ and $-q_2$ is $2(m+n)-2Rank(A)$, which by the Rank-Nullity Theorem, i.e., 
$Rank(A)+Null(A)=n+m$, is also equal to $2Null(A)$.

{\bf Case C. ($\lambda_{i}\neq 0$ and $\xi_{i_1}=1$)}.

In this case, by (\ref{E013}), we have $\lambda_{i}=\pm\sqrt{(1+q_1)(1+q_2)}$. Corresponding to eigenvalue $\xi=1$ of $\widetilde{A}^2$,
we have eigenvalues $\pm 1$ of $A_e$ (see, Lemma~\ref{lem02}). If the multiplicity of $\xi=1$ is $m$, we have $m/2$ eigenvalues $+1$ and 
$m/2$ eigenvalues $-1$ for $A_e$. In Fact 4 that follows, we prove that $m = 2|E|-2|V|+2$ (proof is given in Appendix \ref{A5}).
This together with the $2|V|-2$ eigenvalues $\eta$ of $A_e$ (or $\xi$ of $\widetilde{A}^2$) obtained in Cases A and B, add up to a total number of $2|E|$.

{\bf Fact 4.} The multiplicity of the eigenvalue $\xi=1$ of $\widetilde{A}^2$ is at least $2|E|-2|V|+2$.\footnote{Based on the total number of eigenvalues for $\widetilde{A}^2$, the 
multiplicity of the eigenvalue $\xi=1$ is equal to $2|E|-2|V|+2$.}
%
\end{proof}

%



\begin{ex}
Let $G$ be the complete bipartite graph $K_{m,n}$. It is well-known that the spectrum of $G$ (eigenvalues of $A$) is $\{0^{m+n-2},\sqrt{mn},-\sqrt{mn}\}$. We thus have
$Null(A)=m+n-2$ and $Rank(A)=2$. We use Theorem \ref{T03} to find the eigenvalues of $A_e$. 

Step 1. The only negative eigenvalue  of $A$ is $-\sqrt{mn}$. By solving the quadratic equation  (\ref{E013}) for $ \lambda= -\sqrt{mn}$, 
we obtain two solutions $1$ and $(m-1)(n-1)$. This gives us eigenvalues  $\eta = \pm\sqrt{(m-1)(n-1)}$ for $A_e$, each with multiplicity one.

Step 2. Matrix $A_e$ has also eigenvalues $\pm\sqrt{-(m-1)}$, each with multiplicity $n-Rank(A)/2=n-1$, and eigenvalues $\pm\sqrt{-(n-1)}$, each with 
multiplicity $m-Rank(A)/2=m-1$. 

Step 3. Also, $A_e$ has eigenvalues $\pm 1$, each with multiplicity $mn-(m+n)+1$.

Consequently, 
using Theorem \ref{T02}, we have
\begin{align*}\label{EEE66}
N_{4}&= \frac{2mn-2(m+n)+2 + 2\Big((m-1)(n-1)\Big)^{4/2}  +(2n-2)(1-m)^{4/2} + (2m-2)(1-n)^{4/2} }{2 \times 4}\\
     &= \frac{ \Big((m-1)(n-1)\Big)  +   \Big((m-1)(n-1)\Big)^{2}  + (n-1)(1-m)^{2} +  (m-1)(1-n)^{2} }{4}\\
     &= \frac{  (m-1)(n-1)\Big( 1+ (m-1)(n-1) + (m-1) + (n-1))\Big) }{4}\\
     &= \frac{  (m-1)(n-1)(mn) }{4}\:,  \numberthis
\end{align*}
and
\begin{align*}\label{EEE67}
N_{6}&= \frac{2mn-2(m+n)+2 + 2\Big((m-1)(n-1)\Big)^{3}  +(2n-2)(1-m)^{3} + (2m-2)(1-n)^{3} }{12}\\
&= \frac{  m(m-1)(m-2)n(n-1)(n-2) }{6} \:. \numberthis
\end{align*}
Equations (\ref{EEE66}) and (\ref{EEE67}) are consistent with the results in the literature~\cite{dehghan2018spectrum}.
\end{ex}


\begin{ex}
Consider the tesseract graph, denoted by $Q_4$, and shown in Fig. \ref{fig002}. This graph, also referred to as the $4$-dimensional hypercube, is bipartite.
It is also $4$-regular, and has parameters $m = n = 8$, and $q_1 = q_2 = 3$. The spectrum of $Q_4$ is 
$\{ (-4)^1, (-2)^4, 0^6,  2^4, 4^1\}$. We use Theorem \ref{T03} to find the eigenvalues of $A_e$. From the spectrum of $A$, we have $Null(A)=6$ and $Rank(A)=10$.

Step 1. Matrix $A$ has two negative eigenvalues: $-4$ and $-2$.  By solving (\ref{E013}) for $ \lambda=-4$, we obtain two solutions $ 1$ and $9$. This accounts for 
eigenvalues $\pm 3$ for $A_e$, each with multiplicity one. Also, by solving (\ref{E013})  for $ \lambda=-2$, we obtain two solutions $-1 \pm 2 \sqrt{2} i$, where $i = \sqrt{-1}$. 
This accounts for four eigenvalues  $\pm \sqrt{ -1 \pm 2 \sqrt{2} i }$ for $A_e$, each with multiplicity $4$. 

Step 2. Matrix $A_e$ also has eigenvalues $\pm\sqrt{-3}$, each with multiplicity $n-Rank(A)/2=3$, and eigenvalues $\pm\sqrt{-3}$, each with multiplicity $m-Rank(A)/2=3$
($\pm\sqrt{-3}$, each with multiplicity $6$, in total). 

Step 3. Also, the matrix $A_e$ has the eigenvalues $\pm 1$, each of multiplicity $ |E|-(m+n)+1=17 $.

Now, we use Theorem \ref{T02} to find the number of $4$-cycles in $Q_4$:
\begin{align*}
N_4= \dfrac{2(3)^4 + 8 (-1 + 2 \sqrt{2} i)^2  + 8 (-1 - 2 \sqrt{2} i)^2+ 12 (3)^2+ 34}{8}=24.
\end{align*}
This matches the multiplicity obtained by the backtracking algorithm of \cite{Jeff}. 
\end{ex}

 \begin{figure}[ht]
 	\begin{center}
 		\includegraphics[scale=.4]{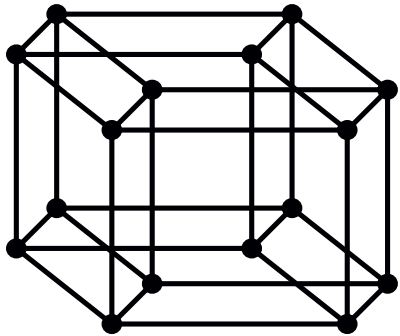}
 		\caption{The tesseract graph $Q_4$.
 		} \label{fig002}
 	\end{center}
 \end{figure}

\section{Conclusion}
\label{section06}
In this paper, we investigated the relationship between the spectra of the adjacency matrix $A$ and the directed edge matrix $A_e$ of a bi-regular bipartite graph.
We proved that the latter spectrum can be derived from the former through simple quadratic equations. Through this relationship, we established a connection between two existing computational methods for counting short cycles (of length less than or equal to $2g-2$, where $g$ is the girth of the graph) in bi-regular bipartite graphs. 
The first method performs such computations using the spectrum of $A_e$ and has complexity $\mathcal{O}(|E|^3)$, where $|E|$ is the number of edges in the graph. The second method uses the graph spectrum and degree sequences of the graph for computations, and has complexity $\mathcal{O}(|V|^3)$,
where $|V|$ is the number of nodes in the graph. The latter complexity can be significantly lower than the former 
for graphs with large node degrees. The downside of the latter approach, however, is that the equations involved in the computations are very tedious, particularly for 
the calculation of multiplicity of $k$-cycles with $k \geq g+6$. Using the results of this work, one can compute the multiplicity of short cycles in a bi-regular bipartite graph using the first approach but with complexity $\mathcal{O}(|V|^3)$ (and without any need for the tedious equations of the second approach).

\section{Appendix} 
\subsection{Proof of Fact 1.}
\label{A2}

We first prove the following lemma which is subsequently used in the proof of Fact 1.

\begin{lem} \label{lem001}
	Let $G=(V=U \cup W,E)$ be a bi-regular bipartite graph with adjacency matrix $A$, and assume that $U=\{u_1,\ldots,u_n\}$ and $W=\{w_1,\ldots,w_m\}$. 
	If $\overrightarrow{\mu}_{i,j}^t=(\mu_{i,j,u_1}, \ldots, \mu_{i,j,u_n},\mu_{i,j,w_1}, \ldots, \mu_{i,j,w_m}  )$ is an eigenvector of $A$ corresponding to the 
	eigenvalue $\lambda_i$ (index $j$ accounts for the possibility of multiple eigenvectors corresponding to the same eigenvalue $\lambda_i$), and $u_k \in U$, then 
	\begin{equation}\label{EEE01}
	\mu_{i,j,u_k}=\dfrac{\sum_{w_t u_k\in E(G)}\mu_{i,j,w_t}}{\lambda_{i}}\:.
	\end{equation}  
\end{lem}

\underline{Proof of Lemma \ref{lem001}}
 
In the adjacency matrix $A$ of the graph $G$, sort the nodes in the following order: $u_1,\ldots,u_n, w_1,\ldots,w_m$. 
Let $ \overrightarrow{u}^t= (\mu_{i,j,u_1}, \ldots, \mu_{i,j,u_n})$ and $ \overrightarrow{w}^t=(\mu_{i,j,w_1}, \ldots, \mu_{i,j,w_m})$. 
Since $\lambda_{i}$ is an eigenvalue of $A$ we have:
	\begin{equation}\label{EE016}
	A
		\left[
	\begin{array}{c }
	 \overrightarrow{u }  \\
	
 \overrightarrow{	w}    \\	
	\end{array}
	\right]=\left[
	\begin{array}{c c}
	0 & D \\
	
	D^t & 0   \\	
	\end{array}
	\right]
	\left[
	\begin{array}{c }
 \overrightarrow{	u}   \\
	
	 \overrightarrow{w}    \\	
	\end{array}
	\right]
	=\lambda_{i}
	\left[
	\begin{array}{c }
 \overrightarrow{	u  } \\
	
 \overrightarrow{	w }   \\	
	\end{array}
	\right]
	\end{equation}
	Thus, $D  \overrightarrow{w}=\lambda_{i}  \overrightarrow{u }$ and $D^t  \overrightarrow{u}=\lambda_{i}  \overrightarrow{w}$. From the first equation, we obtain (\ref{EEE01}).
	This completes the proof of the lemma.

To prove Fact 1, we first show that vectors $\overrightarrow{\phi}_{i,1}, \ldots,\overrightarrow{\phi}_{i,m_i}$ are linearly independent.
To prove the claim, we use contradiction. To the contrary, assume that vectors $\overrightarrow{\phi}_{i,1}, \ldots,\overrightarrow{\phi}_{i,m_i}$ are 
not linearly independent. So, there are constant numbers, $c_{i,1}, \ldots, c_{i, m_i}$, such that at least two are non-zero and we have 
\begin{equation} \label{EE007}
c_{i,1}\overrightarrow{\phi}_{i,1}+ \cdots+ c_{i, m_i}\overrightarrow{\phi}_{i,m_i} =\overrightarrow{0}\:.  
\end{equation}
Let $xy\in E$,  $x\in U$, and $y\in W$. Consider the row corresponding to the pair of nodes $(x,y)$ in (\ref{EE007}). We have: 
\begin{equation} \label{EE008}
c_{i,1}{\phi}_{i,1,(x,y)}+ \cdots+ c_{i, m_i}{\phi}_{i,m_i,(x,y)} =0\:.
\end{equation}
By substituting (\ref{EEE02}) in (\ref{EE008}) and applying $f_1= \frac{\lambda_{i}}{\xi_{i_1}+q_1}$, we have 
\begin{equation}\label{EE009}
\displaystyle\sum_{j=1}^{m_i} c_{i,j}(\mu_{i,j,y}-\frac{\lambda_{i}}{\xi_{i_1}+q_1} \mu_{i,j,x})=0\:.
\end{equation}
Since $q_1 +1 \geq 2$, there is a node $y'\in W$, such that $y'\neq y$ and $xy'\in E(G)$. Similar to (\ref{EE009}), we thus have
\begin{equation}\label{EEE04}
\displaystyle\sum_{j=1}^{m_i} c_{i,j}(\mu_{i,j,y'}-\frac{\lambda_{i}}{\xi_{i_1}+q_1} \mu_{i,j,x})=0\:.
\end{equation}
From (\ref{EE009}) and (\ref{EE004}),  we obtain
\begin{equation}\label{EEE05}
\displaystyle\sum_{j=1}^{m_i} c_{i,j}\mu_{i,j,y}
=
\displaystyle\sum_{j=1}^{m_i} c_{i,j}\mu_{i,j,y'}\:.
\end{equation}
Since the graph is connected, for any two nodes $y,y'\in W$, we have (\ref{EEE05}).
By the same approach, for every two nodes $x,x'\in U$, we have
\begin{equation}\label{EEE0t}
\frac{\lambda_{i}}{\xi_{i_1}+q_1}\displaystyle\sum_{j=1}^{m_i} c_{i,j} \mu_{i,j,x}
= 
\frac{\lambda_{i}}{\xi_{i_1}+q_1}\displaystyle\sum_{j=1}^{m_i} c_{i,j}\mu_{i,j,x'}\:.
\end{equation}
In Case A, we assumed that $\lambda_{i}\neq 0 $. So, by (\ref{E013}), we have $ \xi_{i_1} \neq -q_1$. Thus, $\frac{\lambda_{i}}{\xi_{i_1}+q_1}$ is a nonzero constant number. Hence, by (\ref{EEE0t}), we have
\begin{equation}\label{EEE07}
\displaystyle\sum_{j=1}^{m_i} c_{i,j} \mu_{i,j,x}
= 
\displaystyle\sum_{j=1}^{m_i} c_{i,j}\mu_{i,j,x'}\:.
\end{equation}

Now, consider the left hand side of (\ref{EE009}). By using Lemma \ref{lem001} for the node $x $ and $\mu_{i,j,x}$, we have
\begin{align*} \label{EEE08}
\displaystyle\sum_{j=1}^{m_i} c_{i,j}\Big(\mu_{i,j,y}-\frac{\lambda_{i}}{\xi_{i_1}+q_1} \mu_{i,j,x}\Big)&= \displaystyle\sum_{j=1}^{m_i} c_{i,j}\Big(\mu_{i,j,y}-\frac{\lambda_{i}}{\xi_{i_1}+q_1} \dfrac{\sum_{y'x\in E(G)}\mu_{i,j,y'}}{\lambda_{i}}\Big)&&\\
&=\displaystyle\sum_{j=1}^{m_i} c_{i,j}\Big(\mu_{i,j,y}-\frac{\sum_{y'x\in E(G)}\mu_{i,j,y'}}{\xi_{i_1}+q_1}\Big) \:.&& \numberthis
\end{align*}
By (\ref{EEE05}), we have
\begin{equation}\label{EEE011}
\displaystyle\sum_{j=1}^{m_i} c_{i,j} \sum_{y'x\in E(G)}\mu_{i,j,y'} = \displaystyle\sum_{j=1}^{m_i} c_{i,j} (q_1+1)\mu_{i,j,y}\:.
\end{equation}
By substituting (\ref{EEE011}) in (\ref{EEE08}), we obtain
\begin{align*}
\displaystyle\sum_{j=1}^{m_i} c_{i,j}\Big(\mu_{i,j,y}-\frac{\lambda_{i}}{\xi_{i_1}+q_1} \mu_{i,j,x}\Big)&=\displaystyle\sum_{j=1}^{m_i} c_{i,j} \Big(\mu_{i,j,y}-\frac{ (q_1+1)\mu_{i,j,y}}{\xi_{i_1}+q_1}\Big) && \\
&=\displaystyle\sum_{j=1}^{m_i} c_{i,j} \mu_{i,j,y}\Big(1-\frac{ q_1+1}{\xi_{i_1}+q_1}\Big) && 
\end{align*}
By (\ref{EE009}), we thus have
\begin{equation}\label{eqas}
\displaystyle\sum_{j=1}^{m_i} c_{i,j} \mu_{i,j,y}\Big(1-\frac{ q_1+1}{\xi_{i_1}+q_1}\Big) = 0\:.
\end{equation}
Since $\xi_{i_1}\neq 1$, thus $1-\frac{ q_1+1}{\xi_{i_1}+q_1}\neq 0$. So,
\begin{equation}\label{EEE10}
\displaystyle\sum_{j=1}^{m_i} c_{i,j} \mu_{i,j,y}=0\:.
\end{equation}
By (\ref{EEE10}) and (\ref{EE009}), and since  $\lambda_{i} \neq 0$, we have
\begin{equation}\label{EEE11}
\displaystyle\sum_{j=1}^{m_i} c_{i,j} \mu_{i,j,x}=0\:.
\end{equation}
Consequently, 
\begin{equation}\label{EEE12}
\displaystyle\sum_{j=1}^{m_i} c_{i,j} \overrightarrow{\mu}_{i,j}=0\:.
\end{equation}
This is, however, in contradiction with the eigenvectors $\overrightarrow{\mu}_{i,1}, \ldots, \overrightarrow{\mu}_{i,m_i}$ being linearly independent.
So, the vectors $\overrightarrow{\phi}_{i,1}, \ldots,\overrightarrow{\phi}_{i,m_i}$ are  linearly independent. With the same approach, we can prove that  the vectors 
$\overrightarrow{\rho}_{i,1}, \ldots,\overrightarrow{\rho}_{i,m_i}$ are linearly  independent.

\subsection{Proof of Fact 2.}
\label{A3}

Consider the following adjacency matrix of a bipartite graph $G$: 
\[
A =\left[
\begin{array}{l|l}
0_{ n\times n}  & D_{n\times m} \\
\hline
D_{m\times n}^t          & 0_{ m\times m}    \\	
\end{array}
\right]\:.
\]
We have
\begin{equation}\label{EEE50}
Rank(A)=Rank(D)+Rank(D^t)\:.
\end{equation}
Also, 
\begin{equation}\label{EEE51}
Rank(D)=Rank(D^t)\:.
\end{equation}
From (\ref{EEE50}) and (\ref{EEE51}), we obtain
\begin{equation}\label{EEE52}
Rank(A)=2Rank(D) \:.
\end{equation}
By the Rank-Nullity Theorem for matrix $D$, we have
\begin{equation}\label{EEE53}
Rank(D)+Null(D)=m\:.
\end{equation}
Thus, by (\ref{EEE52}) and (\ref{EEE53}), we have
$$
Null(D)=m-Rank(A)/2\:.
$$
Similarly, 
$$
Null(D^t)=n-Rank(A)/2\:.
$$

\subsection{Proof of Fact 3.}
\label{A4}

We show that the vector $\overrightarrow{\phi}_{i}$  is an eigenvector of $\widetilde{A}^2$ associated with eigenvalue $-q_2$.
Let $\xi \neq 0$ be an eigenvalue of $\widetilde{A}^2$ corresponding to an eigenvector $\overrightarrow{\phi}$.
Then, by Lemma~\ref{lemc}, $\phi_{(x,y)} =0$, for any pair of nodes $(x,y)$, where $x$ and $y$ are on the same side 
of the bipartition. On the other hand, for $(u,w)$, where $u\in U$ and $w\in W$, by (\ref{E004}), we have
\begin{equation}\label{EEE80}
\xi \phi_{(u,w)} =a_{uw}\sum_{x\in U} a_{wx}\sum_{y\in W} a_{xy} \phi_{(x,y)}
-a_{uw}\sum_{x\in U} a_{wx}  \phi_{(x,w)}
-a_{uw}\sum_{y\in W} a_{uy}   \phi_{(u,y)}
+a_{uw}\phi_{(u,w)} \:.
\end{equation} 
By replacing (\ref{EEE56}) in the right hand side of (\ref{EEE80}), we obtain
\begin{equation}\label{eqwh}
a_{uw}\sum_{x\in U} a_{wx}\sum_{y\in W} a_{xy} \mu_y
-a_{uw}\sum_{x\in U} a_{wx}  \mu_w
-a_{uw}\sum_{y\in W} a_{uy}   \mu_y
+a_{uw}\mu_w \:. 
\end{equation}
Now considering that $\overrightarrow{\mu}$ is in the null space of $D$, the summation $\sum_{y\in W} a_{xy} \mu_y$
in the first term of (\ref{eqwh}) and $\sum_{y\in W} a_{uy} \mu_y$ in the third term are zero. The second term of  (\ref{eqwh}) can also be simplified 
to $-a_{uw} (q_2+1) \mu_w$. Thus, Equation (\ref{eqwh}) reduces to $-q_2 a_{uw} \mu_w$, or $-q_2 \phi_{(u,w)}$, where $\phi_{(u,w)}$
is the $(u,w)^{th}$ element of $\overrightarrow{\phi}_{i}$, as shown in (\ref{EEE56}). Similarly, for $(w,u)$, where $u\in U$ and $w\in W$,
by replacing (\ref{EEE56}) in the right hand side of
(\ref{EE001}), and some simplifications, we obtain $-q_2 a_{wu} \mu_w$, which is equal to $-q_2 \phi_{(w,u)}$, where $\phi_{(w,u)}$
is the $(w,u)^{th}$ element of $\overrightarrow{\phi}_{i}$, as shown in (\ref{EEE56}). This completes the proof  that $\overrightarrow{\phi}_{i}$  
is an eigenvector of $\widetilde{A}^2$ associated with the eigenvalue $\xi=-q_2$.

Similarly, it can be shown that  $\overrightarrow{\phi}_{i}'$ is an eigenvector of $\widetilde{A}^2$ associated with eigenvalue $-q_2$.

\subsection{Proof of Fact 4.}
\label{A5}
To prove the result, we use Lemma~\ref{lemc} to characterize the system of linear equations that describe the eigenvectors of 
$\widetilde{A}^2$ associated with the eigenvalue $\xi=1$. 

First, corresponding to each edge $xy\in E(G)$, we define two variables  $\psi_{(x,y)}$ and $\psi_{(y,x)}$, for a total of $2|E|$ variables.
We then define the vector $\overrightarrow{\rho}$ as:
\begin{equation}\label{EEE444}
\rho_{(x,y)}=
\begin{cases}
\psi_{(x,y)},       & \text{if}\,\,xy\in E(G), \\
0            ,       & \text{otherwise}.\
\end{cases}
\end{equation}
Now, for each node $ u\in U$, consider the following two linear equations (involving variables $\psi_{(x,y)}$ and $\psi_{(y,x)}$):
\begin{equation}\label{EEE16}
\sum_{y\in W} a_{uy}   \psi_{(u,y)}=0\:,
\end{equation}
and
\begin{equation}\label{EEE17}
\sum_{y\in W} a_{uy}  \psi_{(y,u)}=0\:,
\end{equation}
and for each node $w\in W$, consider the following two linear equations:
\begin{equation}\label{EEE18}
\sum_{x\in U} a_{wx}  \psi_{(x,w)}=0\:,
\end{equation}
and
\begin{equation}\label{EEE19}
\sum_{x\in U} a_{wx}   \psi_{(w,x)}=0\:.
\end{equation} 
One can see that if we have the above equations (i.e. (\ref{EEE16}) and (\ref{EEE17}) for each $ u\in U$, and (\ref{EEE18}) and (\ref{EEE19}) for each $w\in W$), 
then by (\ref{E004}) and (\ref{EE001}), the vector $\overrightarrow{\rho}$, given in (\ref{EEE444}), is an eigenvector of  $\widetilde{A}^2$ associated with  
eigenvalue $\xi=1$. We note that the total number of equations in (\ref{EEE16}), (\ref{EEE17}), (\ref{EEE18}) and (\ref{EEE19}) is $2|V|$.
From this set of $2|V|$ equations, however, at least two are redundant. To show this, consider Equation (\ref{EEE18}) for a specific node $w\in W$.
This equation can be derived from all the remaining equations in  (\ref{EEE18}), and the following equation:
\begin{equation}\label{EEE40}
\sum_{x\in U} \sum_{  y\in W} a_{xy}   \psi_{(x,y)}=0\:,
\end{equation}
which itself is obtained by adding up equations in (\ref{EEE16}) for all the nodes in $U$. Similarly, one of the equations in (\ref{EEE19}) can be deemed redundant,
as it can be derived from the rest of the equations in (\ref{EEE19}), and the equation obtained by adding up all the equations in (\ref{EEE17}).
Having at least two redundant equations, and removing them from the system of linear equations, we have now $2|V|-2$ linear equations and 
$2|E|$ variables. As a result, we have at least $2|E|-2|V|+2$ linearly independent solutions for the eigenvector $\overrightarrow{\rho}$.

\bibliographystyle{ieeetr}

\end{document}